\documentclass[letterpaper]{article}

\usepackage{dsfont}
\usepackage{amsmath}
\usepackage{amsfonts}
\usepackage{amssymb}
\usepackage{mathtools}

\usepackage{subcaption}

\usepackage{amsthm}

\newtheorem{theorem}{Theorem}
\newtheorem{lemma}{Lemma}

\newenvironment{proofsketch}{%
  \proof}{\endproof}

\usepackage{xcolor}
\usepackage{enumitem}

\usepackage{natbib,alifeconf}  %% The order is important
\usepackage{url,hyperref,cleveref}
\usepackage{booktabs}

\title{Mimicry and the Emergence of Cooperative Communication}

\author{
    Dylan Cope \and
    Peter McBurney
    \mbox{}\\
    King's College London \\
    \url{dylan.cope@kcl.ac.uk}
} % email of corresponding author

\begin{document}

\maketitle

\begin{abstract}
    In many situations, communication between agents is a critical component of cooperative multi-agent systems, however, it can be difficult to learn or evolve.
    In this paper, we investigate a simple way in which the emergence of communication may be facilitated.
    Namely, we explore the effects of when agents can mimic preexisting, externally generated useful signals.
    The key idea here is that these signals incentivise listeners to develop positive responses, that can then also be invoked by speakers mimicking those signals.
    This investigation starts with formalising this problem, and demonstrating that this form of mimicry changes optimisation dynamics and may provide the opportunity to escape non-communicative local optima.
    We then explore the problem empirically with a simulation in which spatially situated agents must communicate to collect resources.
    Our results show that both evolutionary optimisation and reinforcement learning may benefit from this intervention.
\end{abstract}

{\def\thefootnote{}\footnotetext{
    \footnotesize
    Dylan Cope is supported by the UKRI Centre for Doctoral Training in Safe and Trusted AI (EPSRC Project EP/S023356/1).
    Accepted for publication in the proceedings of the 2024 \textit{International Conference on Artificial Life (ALIFE24)}.
}}

% \section{Acknowledgements}
% This work was supported by NSF grant No.\ 123456.

\section{Introduction}

The emergence of communication between cooperative co-evolving agents has a chicken-and-egg problem.
For producing a signal to be advantageous, the speaker needs to invoke a desirable behaviour in the listener.
Likewise, from the listener's perspective, attending to signals is only advantageous if they carry pertinent information.
% In Darwinian terms, different signals need to result in variation in the expected fitness.
Therefore, communication only emerges when, by chance, the speaker happens to signal relevant information and the listener happens to respond positively to those signals.
In this paper, we are interested in how the emergence of communication may be facilitated by the capacity for speakers to mimic externally generated signals that already carry valuable information.
The key hypothesis here is that the external signals provide an incentive for the listener to develop a positive response to a stimulus.
Then, the speaker can mimic that signal to also invoke the same beneficial behaviour in the listener.

To investigate this idea, we start with a theoretical analysis of some of the difficulties that emergent communication faces.
We first focus on independent optimisers and show how the potential for mimicry alters the optimisation criteria to help communication develop.
Then we shift to analysing centralised optimisers.
We look at how, despite ameliorating some issues, communication may still fail to develop due to local optima.
However, here again, the capacity for agents to mimic useful external signals improves the potential for communicative strategies to develop.
Finally, we conduct an empirical investigation using a simulation with spatially situated agents collecting resources.
We train agents to solve this task using deep neuroevolution and multi-agent reinforcement learning, to compare different optimisation methods.
These results demonstrate that mimicable signals can have a positive effect, and mimicry does indeed emerge.

Unlike most studies of mimicry, this investigation concerns \textit{cooperative communication}, i.e. signalling between agents acting to achieve shared goals.
This is in contrast to \textit{deceptive communication}, where a preexisting signalling system is exploited by an adversary that mimics signals to its advantage.
However, communication as a tool for cooperation is indispensable for complex collective decision-making.
Furthermore, \cite{grice_logic_1975} suggests that a \textit{cooperative principle} is fundamental to human language.
Therefore, exploring the evolution of cooperative communication may shed light on the emergence of language, or provide insight into developing sophisticated cooperative artificial agents.

\section{Background and Related Work}

\subsection{Mimicry}

Mimicry is a widely studied phenomenon in evolutionary theory \citep{maran_mimicry_2017, kleisner_introduction_2019, wickler_mimicry_1965, quicke_mimicry_2017}, tracing back to the work of \cite{bates_contributions_1862}.
This form of mimicry, called \textit{Batesian mimicry}, involves a system with three parts: a \textit{model} that a \textit{mimic} takes on the perceived characteristics of, to fool a \textit{perceiver} \citep{kleisner_introduction_2019}.
Deceptive communication can take the form of \textit{aggressive mimicry} where a predator exploits a prey, for instance, some spiders lure male moths by mimicking female moth sex pheromones \citep{yeargan_juvenile_1996}.
Alternatively, \textit{defensive mimicry} involves prey exploiting predators, such as non-toxic moths imitating the sounds produced by toxic moths to deter bats \citep{oreilly_deaf_2019}.
In both cases, the mimic generates a signal to trick or deceive another agent.
Researchers in Artificial Life on mimicry have expanded on this work, with deceptive communication studied in ecosystem simulations \citep{hraber_ecology_1997, lehman_surprising_2018,islam_modeling_2016,marriott_social_2018,reynolds_interactive_2011,floreano_evolutionary_2007}.

\subsection{Genetic Algorithms and Neuroevolution}

Genetic Algorithms (GAs) are a class of optimisation algorithms that evolve a population of `genomes' that encode (attempted) solutions to a given task \citep{mitchell_introduction_1996,holland_genetic_1992,floreano_bio-inspired_2008}.
Performance on the task, called fitness, is measured for each member of the population, and then selection and repopulation processes are applied to construct the next generation of the simulation.
The most basic selection method is \textit{truncation}, where a subset of the best individuals are used for repopulation.
The population is restored to its original size by producing \textit{off-spring}.
This may be done by applying random perturbations to parents (i.e. \textit{mutations}), and/or by \textit{sexual recombination} (also called crossover) where segments of two or more genomes are combined.
For almost as long as there have been GAs and artificial neural networks, there have been approaches for optimising network parameters with GAs \citep{ronald_genetic_1994, angeline_evolutionary_1994,stanley_evolving_2002,galvan_neuroevolution_2021}.
In deep reinforcement learning, simple GAs can be effective \citep{such_deep_2017}, and population-based methods have been demonstrated to soften the exploitation-exploration problem \citep{conti_improving_2018,jaderberg_human-level_2019}.

\subsection{Decentralised POMDPs}

A \textit{Decentralised Partially-Observable Markov Decision Process} (Dec-POMDP) is a formal model of a cooperative, multi-agent, sequential decision-making problem defined as a tuple  $\mathcal{M} =(\mathcal{S},\mathcal{A},T,r,\boldsymbol{\Omega},O)$ \citep{oliehoek_concise_2016}, where $\mathcal{S}$ is a set of states, and $\mathcal{A} = \prod_i \mathcal{A}^{i}$ is a product of individual agent action sets.
A \textit{joint action} $\mathbf{a} \in \mathcal{A}$ is a tuple of actions from each agent that is used to compute the environment's transition dynamics, defined by a probability distribution over states $T:\mathcal{S}\times\mathcal{A}\times\mathcal{S}\rightarrow[0,1]$.
Team performance is defined by a cooperative reward function $r: \mathcal{S} \times \mathcal{A} \times \mathcal{S}$ over state transitions and joint actions. $\boldsymbol{\Omega} = \{\Omega^{i}\}$ is a set of observation sets, and $O:S\rightarrow\prod_i \Omega^i $ is an observation function.
For the purposes of this paper, each agent $i$ follows a policy $\pi^i_\theta$ that maps observation sequences to action distribution, where $\theta$ is some parameterisation that may optimised (using methods that we will discuss in the following subsections).

% A \textit{trajectory} for an agent $i$ is a sequence of observation-action-reward tuples $\tau_i \in \mathcal{T}_i = (\Omega_i\times\mathcal{A}_i\times\mathds{R})^*$. For a set of policies $\Pi = \{\pi_i\}$, a joint trajectory is $\boldsymbol{\tau} \in \mathcal{T} = (\boldsymbol{\Omega}\times\mathcal{A}\times\mathds{R})^*$, and we can denote the distribution of joint trajectories for this set of policies acting in the environment as $\mathcal{M}|_{\Pi}$.
% In this work, we will only consider finite-horizon Dec-POMDPs, so the lengths of trajectories will always be bounded.
% The \textit{total reward} for this trajectory is the sum of rewards along the sequence, denoted $R(\boldsymbol{\tau})$.
% The expected sum of rewards for a set of policies will be denoted $R(\Pi) = \mathds{E}_{\boldsymbol{\tau} \sim \mathcal{M}|_{\Pi}}[R(\boldsymbol{\tau})]$. 
% We denote the distribution over future trajectories for a policy as $\pi(\tau|.)$. The \textit{return} of a trajectory is computed as the discounted sum of rewards: $V(\tau) = \sum_{k = 0}^{|\tau|} \gamma^k r_k$. 

\subsection{Emergent Communication}

Emergent communication is the study of agents that learn or evolve to communicate.
Each agent's action set in the Dec-POMDP can be expressed as $\mathcal{A}_i = \mathcal{A}_i^e \times \mathcal{A}^c_i$ or $\mathcal{A}_i = \mathcal{A}_i^e \cup \mathcal{A}^c_i$, where $\mathcal{A}_{i}^c$ is a set of \textit{communicative actions}, and $\mathcal{A}_i^e$ is a set of \textit{environment actions}.
% Typically, the communicative actions can further be written as the product of one-way communication channels from $i$ to $j \in C_i$ using a discrete \textit{message} alphabet $\Sigma$, i.e. $\mathcal{A}^c_i=\Sigma^{|C_i|}$.
% These are cheap-talk channels, meaning there is no cost to communication.
This variant of a Dec-POMDP is known as a Dec-POMDP-Com \citep{goldman_decentralized_2004, goldman_communication-based_2008, oliehoek_concise_2016}.
The messages have no prior semantics as the transition function of the Dec-POMDP only depends on the environment actions ${\mathcal{A}_i^e}$, and agents are not a priori programmed to send messages with predetermined meanings.
Therefore, semantics emerge through maximising reward.

% The evolution of communication the subject of a large variety of Artificial Life  research 
Many works in Artificial Life research have studied the evolution of communication \citep{ackley_altruism_1994,oliphant_dilemma_1996,bullock_evolutionary_1997,parisi_artificial_1997,noble_adaptive_2002,floreano_evolutionary_2007,mirolli_evolving_2010,fox_nectar_2023}, with a variety of approaches and settings considered.
However, to the best of our knowledge this is the first study of the effects of mimicry on the emergence of communication in a cooperative setting.

\subsection{Deep Multi-Agent Reinforcement Learning}

Reinforcement learning (RL) is a paradigm of machine learning algorithms generally formulated in terms of solving Markov decision processes \citep{sutton_reinforcement_1998}.
Advances in deep learning over the last decade have been applied to RL in the form of value estimation \citep{mnih_human-level_2015,hessel_rainbow_2018} and policy gradient optimisation \citep{schulman_proximal_2017,lillicrap_continuous_2016,wang_sample_2017,schulman_trust_2015}.
For this work, we will use the Multi-Agent Proximal Policy Optimisation (MAPPO) algorithm \citep{yu_surprising_2022}, a centralised multi-agent policy gradient method for training cooperative agents.

When applied to multi-agent communication problems, deep multi-agent RL (DMARL) enables the study of more complex behaviours.
However, a fundamental issue with multi-agent learning problems is that when each agent learns independently, the learning dynamics may be chaotic \citep{hussain_asymptotic_2023,hussain_stability_2024}.
This is especially tricky for communication, where the only thing that grounds a message's consequences are the effects on the receiver's actions.
\cite{foerster_learning_2016} and \cite{sukhbaatar_learning_2016} were the first to approach this issue by developing differentiable communication channels between agents that can be discretised after training.

% \section{Emergent Communication with \\ Independently Optimising Agents}

\section{Independently Optimised Communication}

In this section, we outline the difficulty for two independent optimising agents to develop communication with one another.
We will look at each agent as an optimiser that maximises a utility function where the other agent is treated as static.
Hence, in this case, the `chicken-and-egg' issue is a matter of each agent optimising non-stationary criteria in which the other agent appears to be random.

Consider a listener $\pi^l_{\theta^l}$ (a policy parameterised by $\theta^l$) that observes signals $m \in \Sigma$ and takes actions $a \sim \pi^l_{\theta^l}(m)$.
These actions result in some utility $U(a|z)$, where $z \in \mathcal{Z}$ is some latent variable.
The agent will optimise $\theta^l$ to maximise the expected utility. 
Assuming that $U$ is known, we define the expected utility for taking action $a$ when observing $m$:
\begin{align}\label{eqn:listener_eu}
    E[U(a|m)] = \sum_{z\in \mathcal{Z}} P(z | m) U(a|z)
\end{align}
In settings where a speaker and listener are co-evolving or co-learning, the initial signals generated by the speaker are random.
Therefore, messages are independent of the latent variable and $P(z | m) = P(z)$.
As a result, when optimising the listener for this task the signals provide no information.
In the worst case, the quantity
\begin{align}\label{eqn:listener_eu_with_uninformative_msgs}
\sum_{z\in \mathcal{Z}} P(z) U(a|z)
\end{align}
is the same for all actions, meaning the listener cannot do better than taking random actions.

Next, we can construct a similar dilemma from the speaker's perspective.
The speaker $\pi^s_{\theta^s}$, a policy parameterised by $\theta^s$, observes the private variable $z \in \mathcal{Z}$ and produces some signal $m \sim \pi^s_{\theta^s}(z)$.
The expected utility for producing a given message is dependent on the listener, so:
\begin{align}\label{eqn:speaker_eu}
    E[U(m|z)] = \sum_{a \in \mathcal{A}} \pi^l_{\theta^l}(a|m) U(a|z)
\end{align}
But again, at first, the listener produces random actions. So when optimising the expected utility for the speaker, we are effectively optimising the following expression,
\begin{align}\label{eqn:speaker_eu_random_listener}
    \frac{1}{|\mathcal{A}|}\sum_{a \in \mathcal{A}} U(a|z),
    \quad\text{as}~\forall a,~ \pi^l_{\theta^l}(a|m)=\frac{1}{|\mathcal{A}|}
\end{align}
which is independent of the speaker's message.
In other words, the speaker cannot affect the expected utility.

\subsection{Externally Generated Useful Signals}\label{sec:external_sources}

To introduce the possibility of mimicry we will introduce \textit{externally generated} signals that already carry useful information.
In this setting, sometimes the listener observes a signal generated by the speaker, otherwise, it comes from an external source.
Importantly, the listener does not know which source each signal has come from.

Let the proposition $\mathrm{S}$ denote that a signal is generated by the external source.
We can decompose $P(z | m)$ from Equation \ref{eqn:listener_eu} by into a weighted sum of these two cases:
\begin{align}
    P(z | m) = P(z | m, \mathrm{S}) P(\mathrm{S}) + P(z | m, \neg \mathrm{S}) P(\neg \mathrm{S})
\end{align}
The speaker being untrained means $P(z | m, \neg \mathrm{S}) = P(z)$.
Making the substitution into Equation \ref{eqn:listener_eu}:
% ORIGINAL ====
% \begin{align}
%     E[U(a|m)] = \sum_{z\in \mathcal{Z}} \Big(&P(z | m, E) P(E) \\&+ P(z) P(\neg E)\Big) U(a|z)
% \end{align}
% 
% JUST WITH NUMBER FIXED ====
% \begin{equation}
%     \begin{aligned}
%         E[U(a|m)] = \sum_{z\in \mathcal{Z}} \Big(&P(z | m, E) P(E) \\&+ P(z) P(\neg E)\Big)U(a|z)
%     \end{aligned}
% \end{equation}
% 
% AS I THINK OPTIMAL ===
\begin{equation}
    \begin{aligned}
        E[U(a|m)] = \sum_{z\in \mathcal{Z}} \Big[\Big(&P(z | m, \mathrm{S}) \cdot P(\mathrm{S}) \\
        &+ P(z) P(\neg \mathrm{S})\Big) \cdot U(a|z)\Big]
    \end{aligned}
\end{equation}
% 
% \begingroup\makeatletter\def\f@size{7}\check@mathfonts
% \begin{equation}
%     \begin{aligned}
%         E[U(a|m)] = \sum_{z\in \mathcal{Z}} \Big[\Big(P(z | m, \mathrm{S}) \cdot P(\mathrm{S})
%         + P(z) P(\neg \mathrm{S})\Big) \cdot U(a|z)\Big]
%     \end{aligned}
% \end{equation}
% \endgroup

Thus, supposing that $P(\mathrm{S})$ is non-zero and $P(z|m,\mathrm{S}) \neq P(z)$, the listener is incentivised to modify its actions depending on the signal it observes.
As a result, the expected utility for the speaker, defined by Equation~\ref{eqn:speaker_eu}, does not reduce to Equation~\ref{eqn:speaker_eu_random_listener} as the listener no longer produces random actions.
Further, this means that is that the speaker is incentivised to mimic the externally generated signals, as those are what invoke the useful behaviours in the listener.

It is worth noting that a balance needs to be struck between giving both the listener and the speaker opportunities to improve.
When $P(E)$ is high, the speaker will emit many signals that are never received.
This is an issue for optimisation methods such as genetic algorithms or reinforcement learning that rely on estimating the expected utility from samples.
As $P(E)$ increases, there will be greater variance in the expected utility estimates and progress will be slow (or impossible).
This implies that this form of mimicry may provide the most advantage when the external source of information is removed at some point or fades over time.
However, we will leave that situation for future work.

\section{Getting Stuck in Local Optima}

Many forms of optimisation are \textit{centralised}, meaning that the problem of non-stationarity can be resolved.
Yet, communication may fail to emerge if the optimisation process gets stuck in a local optimum.
In the following analysis, we will consider a simple set-up in which both communicative and non-communicative strategies are possible.
The aim is to show that even when communicative strategies are globally optimal, escaping local non-communicative optima may require a significant jump in communicative capabilities.
Thereby making the transition from non-communicative to communicative behaviour unlikely.

\subsection{To Speak, or Not to Speak}

To make our analysis more concrete, we will construct a simple problem where communication is optional by defining a utility function $U$.
In the previous section, we considered a setting where the only choice for the speaker was which signal to send.
Here, the speaker chooses from a set of actions $\mathcal{A}^s$, where only a subset of these actions are signals: $\Sigma \subset \mathcal{A}^s$.
The rest of the actions represent the speaker choosing to not communicate and acting independently of the listener.
When the speaker does not send a signal, the listener observes a `silent' symbol.

However, this setting is still cooperative, so these actions will still contribute some expected utility to the group.
We will introduce a simple utility function $U$ composed of $U^s$, a function of the speaker's actions, and $U^l$, a function of the listener's actions.
\begin{equation}
    \begin{aligned}
        &U(a^s, a^l | z) = U^s(a^s) + U^l(a^l|z), 
        \quad \text{where:}~\\
        &U^s(a^s) = \begin{cases}
            0 & \text{if}~a^s \in \Sigma \\
            u & \text{otherwise}
        \end{cases}
        ~\quad
        U^l(a^l|z) = \begin{cases}
            1 & \text{if}~a^l = z \\
            0 & \text{otherwise}
        \end{cases}
    \end{aligned}
\end{equation}
In these equations, $a^s\in \mathcal{A}^s$ is the speaker's action, $a^l\in \mathcal{A}^l$ is the listener's action, and $z \in \mathcal{Z}$ is a variable hidden from the listener and known to the speaker.
This variable is drawn from a uniform random distribution $z \sim U(\mathcal{Z})$.
We will assume that $|\mathcal{A}^l| = |\mathcal{Z}|$, so each of the listener's actions corresponds to a guess about which $z$ the speaker observed.
In our notation, $a^l = z$ indicates that a correct guess was made, which awards the team 1 utility.

The speaker's direct contribution to the team, $U^s$, is defined by whether it communicates.
If a communicative action is chosen, i.e. the action is a signal $a^s \in \Sigma$, and no utility is awarded.
In this case, the team's utility solely derives from the listener's ability to guess the correct answer.
When the speaker chooses not to communicate, i.e. $a^s \notin \Sigma$, the team is given a fixed reward $u$.
This implicitly adds a cost to communication, as the speaker must give up $u$ reward to send a signal.

\subsection{When is Communication Selected?}

We will denote the parameters for the speaker and listener policies together as $\theta$, as we assume that both agents are jointly optimised.
Given some $\theta$, we can express the expected utility as:
\begin{align}\label{eqn:exp_utilty_theta}
    E[U~|~\theta] = u P(a^s \notin \Sigma~|~\theta) + P(a^l = z~|~\theta)
\end{align}
We will refer to $P(a^l = z~|~\theta)$ as the \textit{guess accuracy}.
Given $|\Sigma| = |A^l|$, it is always possible to successfully communicate.
So when searching for parameterisations that maximise the expected utility we can consider two cases:
\begin{align}
    \max_{\theta}E[U~|~\theta] = \max \Big\{
        u + \frac{1}{|\mathcal{Z}|}, 1
    \Big\}
\end{align}
The first case represents the strategy where the speaker never communicates, therefore the team always gets $u$ reward, and the listener makes the correct guess one in $|\mathcal{Z}|$ times.
We will denote this quantity as $\alpha$ and refer to this strategy as the \textit{optimal non-communicative strategy}.
Similarly, the second case is the \textit{optimal communicative strategy}, in which the speaker always signals the correct answer to the listener.
Therefore, whether communication is the optimal overall strategy is given by:
\begin{align}\label{eqn:comm_optimal_cond}
    \alpha < 1,\quad \text{where}~\alpha = u + \frac{1}{|\mathcal{Z}|}
\end{align}
The condition in Equation~\ref{eqn:comm_optimal_cond} only tells us that communication is \textit{globally optimal}.
To explore local optima, we will consider a simple greedy population-based optimiser, for instance, a genetic algorithm.
At each optimisation step, the optimiser selects the member of the $\theta_i$ that has the highest $E[U|\theta_i]$.
Suppose we have two possible parameterisations, $\theta_1$ and $\theta_2$ (i.e. a population of size two).
We are interested in the case where one of the parameterisations corresponds to the speaker being more likely to communicate successfully:
\begin{equation}\label{eqn:t2_prefers_comm}
    \begin{aligned}
        P(a^s\in \Sigma ~|~ \theta_2) &> P(a^s\in \Sigma ~|~ \theta_1) \\
        P(a^l = z~|~a^s\in \Sigma, \theta_2) &> P(a^l = z~|~a^s\in \Sigma, \theta_1)
    \end{aligned}
\end{equation}
The first condition states that the speaker under $\theta_2$ is more likely to send signals.
The second condition states that when this happens, the listener is more likely to guess correctly, i.e. communication is more often successful.
The optimiser will select the more communicative parameterisation $\theta_2$, and thus we will say that communication is \textit{locally optimal}  if:
\begin{align}\label{eqn:t2_selection_criteria}
    E[U~|~\theta_2] > E[U~|~\theta_1]
\end{align}
Now we may ask the following question: given that a communicative strategy competes against an optimal non-communicative strategy, how high does the guess accuracy need to be for communication to be selected?

This is important because if any increase in accuracy results in the communicative strategy being selected, then any small amount of communication that emerges will be selected by the optimiser.
This would allow communication to emerge gradually as a shift away from optimal non-communicative strategies.
However, we will now show that in this competition there is a significant gap that needs to be overcome for communication to emerge, and therefore it cannot happen incrementally.

\begin{theorem}\label{thm:local_noncomm_optima_gap}
    Given a $\theta_1$ that implements an optimal non-communicative strategy, and a communicative strategy $\theta_2$.
    Under $\theta_2$, when the speaker sends a signal, the guess accuracy must be greater than $\alpha$ for $\theta_2$ to be selected.
    Formally, the following must hold for $\theta_2$ to be selected:
    \begin{align}\label{eqn:local_noncomm_optima_gap}
         P(a^l = z | a^s \in \Sigma, \theta_2) > \alpha
    \end{align}
\end{theorem}
\begin{proof}
% \textbf{Proof.}
We will first rewrite the expected utility as:
\begin{align}\label{eqn:exp_utilty_theta2}
    \begin{aligned}
    E[U~|~\theta] &= u P(a^s \notin \Sigma~|~\theta) + P(a^l = z~|~\theta)
    \end{aligned}
    \\
    \begin{aligned}
    &= u P(a^s \notin \Sigma~|~\theta) \\
    &+ P(a^l = z | a^s \notin \Sigma, \theta) P(a^s \notin \Sigma~|~\theta) \\
    &+P(a^l = z | a^s \in \Sigma, \theta) P(a^s \in \Sigma~|~\theta)
    \end{aligned}
\end{align}

In words, the expected utility is broken down into three terms by applying the law of total probability to the second term of the first line.
This breaks the probability that the listener gets the correct answer into the following cases:
\begin{enumerate}[noitemsep]
    \item The speaker did not send a signal, but the listener got the correct answer anyway.
    Regardless of how the listener acts, there is always a one in $|\mathcal{Z}|$ probability that they get the correct answer. Therefore, 
    \begin{align}\label{eqn:random_listener_correct_prob}
    P(a^l = z | a^s \notin \Sigma, \theta) = \frac{1}{|\mathcal{Z}|}, ~\forall \theta
    \end{align}
    \item The speaker did send a signal, and the quantity $P(a^l = z | a^s \in \Sigma, \theta_2)$ denotes the (improved) accuracy of the listener.
\end{enumerate}
Substituting Equation~\ref{eqn:random_listener_correct_prob} into the equation for $E[U~|~\theta]$, we can group the terms that correspond to the speaker not signalling: 
\begin{align}\label{eqn:exp_utilty_theta2}
    \begin{aligned}
    E[U~|~\theta] =~& \Big(u + \frac{1}{|\mathcal{Z}|}\Big) P(a^s \notin \Sigma~|~\theta) \\ 
    &+ P(a^l = z | a^s \in \Sigma, \theta) P(a^s \in \Sigma~|~\theta)
    \end{aligned}
    \\
    \begin{aligned}
     =~& \alpha P(a^s \notin \Sigma~|~\theta) \\ 
    &+ P(a^l = z | a^s \in \Sigma, \theta) P(a^s \in \Sigma~|~\theta)
    \end{aligned}
\end{align}

Therefore, we can substitute this expression for $E[U~|~\theta_2]$ into the selection criteria in Equation~\ref{eqn:t2_selection_criteria} to find:
\begin{align}\label{eqn:t2_selection_criteria_2}
P(a^l = z | a^s \in \Sigma, \theta_2) > \alpha + \frac{E[U|\theta_1] - \alpha}{P(a^s \in \Sigma~|~\theta_2)}
\end{align}
The LHS of this inequality is the probability that the listener gets the correct answer, given that the speaker sent a signal.
So this measures how successfully the pair can communicate under $\theta_2$.
If we hold $\alpha$ constant, the RHS is a function of the probability that the speaker chooses to send a signal under $\theta_2$, and the expected utility of the other parameterisation $E[U|\theta_1]$.
For Theorem~\ref{thm:local_noncomm_optima_gap} we have assumed that $\theta_1$ corresponds to some speaker/listener pair that does not communicate, thus:
\begin{align}\label{eqn:bound_on_noncomm_eu}
    \frac{1}{|\mathcal{Z}|} \leq E[U|\theta_1] \leq \alpha
\end{align}
The lower bound of this inequality corresponds to a speaker that always tries to communicate with a listener that randomly guesses, this is the \textit{poorest performing non-communicative strategy}. 
The upper bound corresponds to a speaker that never sends a signal, so this is the optimal non-communicative strategy.
Applying the upper bound to Equation~\ref{eqn:t2_selection_criteria_2}, the second term disappears and we are left with Equation~\ref{eqn:local_noncomm_optima_gap}, thereby demonstrating Theorem~\ref{thm:local_noncomm_optima_gap}.
% \qedsymbol{}
\end{proof}

\begin{theorem}\label{thm:local_comm_optima}
    Given that $\theta_1$ implements the poorest performing non-communicative strategy, any communicative strategy $\theta_2$ can be selected.
\end{theorem}
\begin{proof}
% \textbf{Proof.}
Applying the lower bound from Equation~\ref{eqn:bound_on_noncomm_eu} to Equation~\ref{eqn:t2_selection_criteria_2}, for the poorest performing non-communicative strategy we find:
\begin{align}\label{eqn:t2_selection_criteria_low_bound}
P(a^l = z | a^s \in \Sigma, \theta_2) > \frac{1}{|\mathcal{Z}|} + u\Bigg(1 - \frac{1}{P(a^s \in \Sigma~|~\theta_2)}\Bigg)
\end{align}
We known that even by guessing randomly the accuracy is at least one in $|\mathcal{Z}|$, therefore we can set $P(a^s \in \Sigma~|~\theta_2) = 1$.
As $\theta_1$ is the poorest performing, $P(a^l = z | a^s \in \Sigma, \theta_1) = \frac{1}{|\mathcal{Z}|}$, and thereby the second assumption given Equation~\ref{eqn:t2_prefers_comm}, the condition in Equation~\ref{eqn:t2_selection_criteria_low_bound} must hold.
Intuitively, as the poorest performing is ignoring the benefit of $u$ for no improvement in listener accuracy, and improvement in accuracy with $\theta_2$ gives it the advantage.
Thus, the communicative strategy must be locally optimal in this situation.
\end{proof}

To summarise, in this subsection, we have shown that given a $\theta_1$ and $\theta_2$ where $\theta_1$ is an optimal non-communicative strategy, and $\theta_2$ is a communicative strategy, $\theta_2$ will only be selected by a greedy optimisation algorithm if the increase in the listener's accuracy is greater than the non-zero constant $\alpha$.
Put simply, to escape the non-communicative local optimum at $\theta_1$ in a single optimisation step the communicative policy needs to make a large improvement to the listener's accuracy.
This jump may be unlikely, as the agents have to go from no communication to a high-performing level of communication in a single step.

\subsection{The Consequences of Listener Competency}

So far, we have explored different conditions in which communicative strategies will be selected.
We have treated both the speaker and listener as jointly parameterised by some $\theta$, but let us now separately consider the speaker parameters $\theta^s$ and listener parameters $\theta^l$.
In this section, we look at how the listener's sensitivity to different messages impacts the potential for the speaker to evolve communication.

Ultimately, we are interested in this investigation to see how being stuck in local optima with different listeners changes the possibility for evolution to escape the local optima.
This in turn is relevant as we will argue that the presence of mimicable external signals will produce more beneficial listeners.
To progress, this analysis will focus on the case of evolution, rather than reinforcement learning, to make the arguments more concrete.
However, an analogous line of reasoning could be drawn for reinforcement learning.

\begin{figure*}
    \centering
    \captionsetup[subfigure]{justification=centering}
    \begin{subfigure}{.38\textwidth}
        \centering
        \includegraphics[height=5cm]{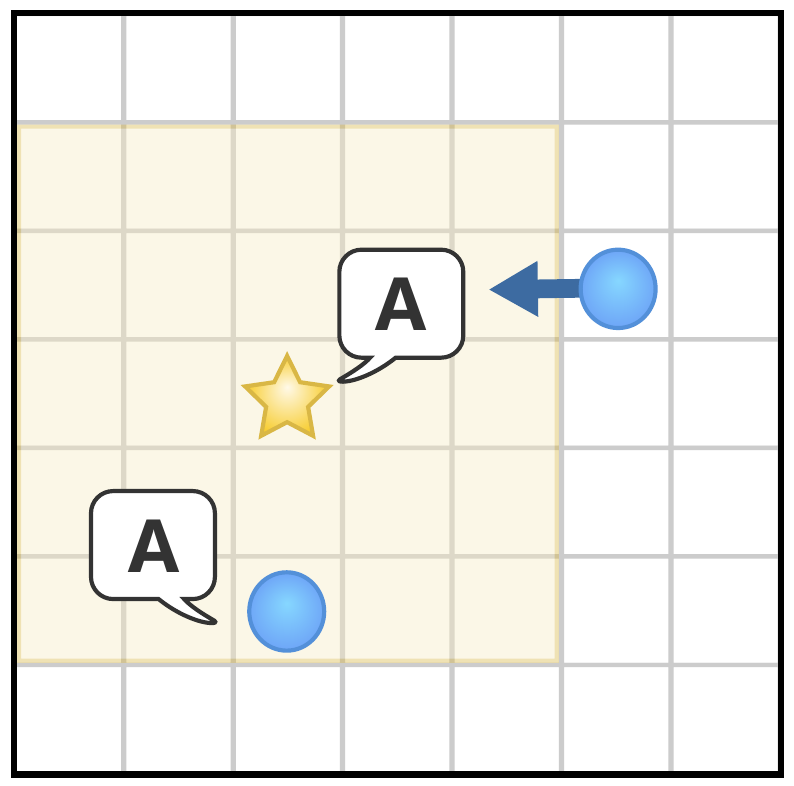}
        \caption{
            Illustration of the environment with two agents (blue circles) and one resource (gold star). \\
            One agent mimics the resource.
        }
        \label{fig:toy-environment-illustration}
    \end{subfigure}%
    \begin{subfigure}{.62\textwidth}
        \centering
        \includegraphics[height=5cm]{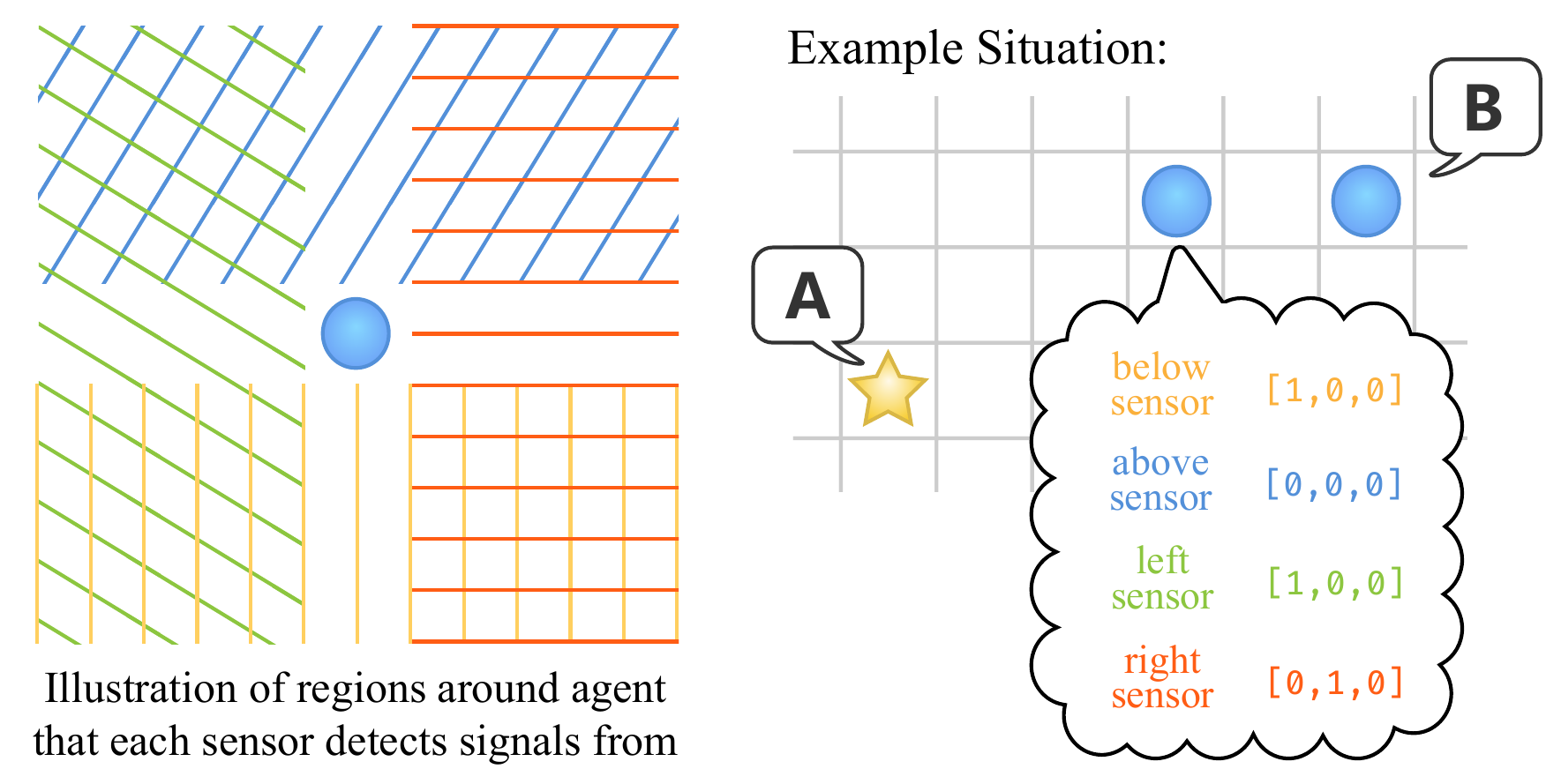}
        \caption{
            Diagram explaining the agents' signal detection mechanisms.
            Four detectors indicate which spatial region relative to the agent a signal was emitted from.
            Therefore, the agent can tell the direction a signal came from, but not the identity of the sender.
        }
        \label{fig:signal-sensor-diagram-full-explanation}
    \end{subfigure}%
    \caption{
        The environment for experimentally testing the effects of mimicry on the emergence of communication.
        Two agents must be on the same square as a resource to collect it and receive a reward.
        Agents observe whether or not they are on the same square as a resource, and resources sometimes emit signals that can be detected by an agent within a limited number of tiles from the resource.
        In (a) the gold region indicates where the signal can be detected by an agent. 
    }
\end{figure*}

We will start by comparing two cases involving different listener strategies.
First, a listener parameterised by $\theta^l_U$ that selects actions uniformly at random, independently of the message it receives:
\begin{equation}
\begin{gathered}
    \forall m \in \Sigma,~P(a^l|a^s=m,\theta^l_U) = P(a^l|\theta^l_U) = \frac{1}{|\mathcal{Z}|}
    % \forall m \in \Sigma,~P(a^l|a^s=m,\theta^l_U) = P(a^l|\theta^l_U)
    % \\
    % \text{and}~P(a^l|\theta^l_U) = \frac{1}{|\mathcal{Z}|}
\end{gathered}
\end{equation}
Secondly, we will consider a deterministic, `competent' listener $\theta^l_I$ that selects some unique action for each message it receives.
We can denote this with a deterministic bijective function $f:\Sigma\rightarrow\mathcal{A}^l$, using the indicator function $\mathds{1}$:
\begin{align}
    P(a^l|a^s=m,\theta^l_U) = \mathds{1}[f(m) = a^l]
\end{align}

Next, suppose we have a population of speakers $\theta^s_i$ that are subject to selection.
For our two cases, the important difference between speakers is the effects on the listener's guess accuracy.
Put differently, for communication to evolve there needs to be variation in the expected $U^l$ between $\theta^s_i$:
\begin{align}
\text{Var}_{\theta^s_i}[E[U^l|\theta^s_i,\theta^l]] > 0
\end{align}

\begin{lemma}
    For $\theta^l_U$, $\text{Var}_{\theta^s_i}[E[U^l|\theta^s_i,\theta^l_U]] = 0$
\end{lemma}
\begin{proofsketch}
\begin{align}
    P(a^l = z | \theta^l_U) = \frac{1}{|\mathcal{Z}|},
    ~\therefore E[U^l|\theta^s_i,\theta^l_U]= \frac{1}{|\mathcal{Z}|}
\end{align}
As the expected utility is constant with respect to $\theta^s_i$, the variance must equal zero.
\end{proofsketch}

With this in mind, we can return to the problem of being stuck in local optima.
Theorem~\ref{thm:local_noncomm_optima_gap} tells us that the guess accuracy (conditioned on signalling) for a communicative strategy needs to be greater than $\alpha$ for the strategy to be selected over a non-communicative optimal strategy $\theta_1$.

For the case of a realistic evolving population at this local optima $\theta_1$, most of the variation in the population will be small changes to the conditional action distributions of the agents.
Given that we are assuming that this population has converged on the local optima, the genetic diversity may be relatively low as many genes have reached fixation in the population.
For simplicity, suppose that the listener's strategy is fixed throughout the population, and we can consider selection against the uniform random listener $\theta^l_U$ or the competent listener $\theta^l_I$.
Equation~\ref{eqn:local_noncomm_optima_gap} only applies when the speaker sends a signal (the probability is conditioned on $a^s \in \Sigma$), so we can focus on those cases.

Now suppose that mutation or recombination events have probability $p$ of transferring at least $\alpha$ probability mass onto some action $a$ for any of the agents' conditional action distributions.
So for instance, before a mutation, a speaker under $\theta_1$ will always choose some non-communicative action $a'$ when they observe $z$, so $P(a'|z)=1$.
But after the mutation, some other, arbitrary action $a$ could have probability $\alpha < P(a|z)$, and $P(a'|z) < \alpha$.
For the competent listener, we replace the probability that $a^l = z$ with $f(a^s) = z$, as $f$ is deterministic and bijective.
This means that whether the condition in Equation~\ref{eqn:local_noncomm_optima_gap} holds is dependent on if the speaker selects the correct signal.
Given that a mutation as described above may shift $\alpha$ probability mass onto any of the signals, the chance of getting a mutation that produces the correct behaviour is $p / |\mathcal{A}^s|$.

\begin{figure*}
    \centering
    \captionsetup[subfigure]{justification=centering}
    \begin{subfigure}{.425\textwidth}
        \centering
        \includegraphics[height=4.25cm]{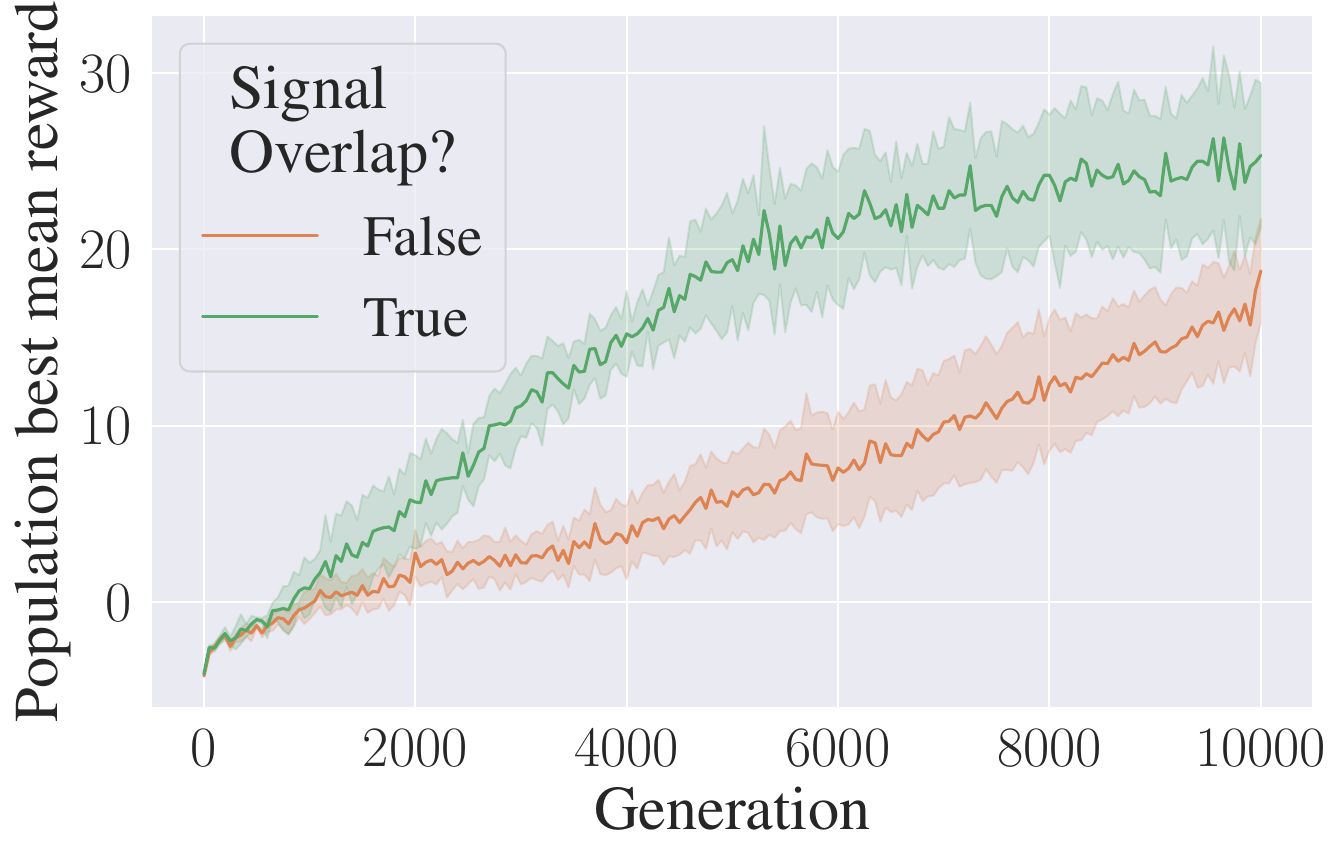}
        \caption{
            Mean reward for the best individual in\\the population at each generation.
        }
        \label{fig:evo_reward_curves}
    \end{subfigure}%
    \hspace{.05cm}
    \begin{subfigure}{.425\textwidth}
        \centering
        \includegraphics[height=4.25cm]{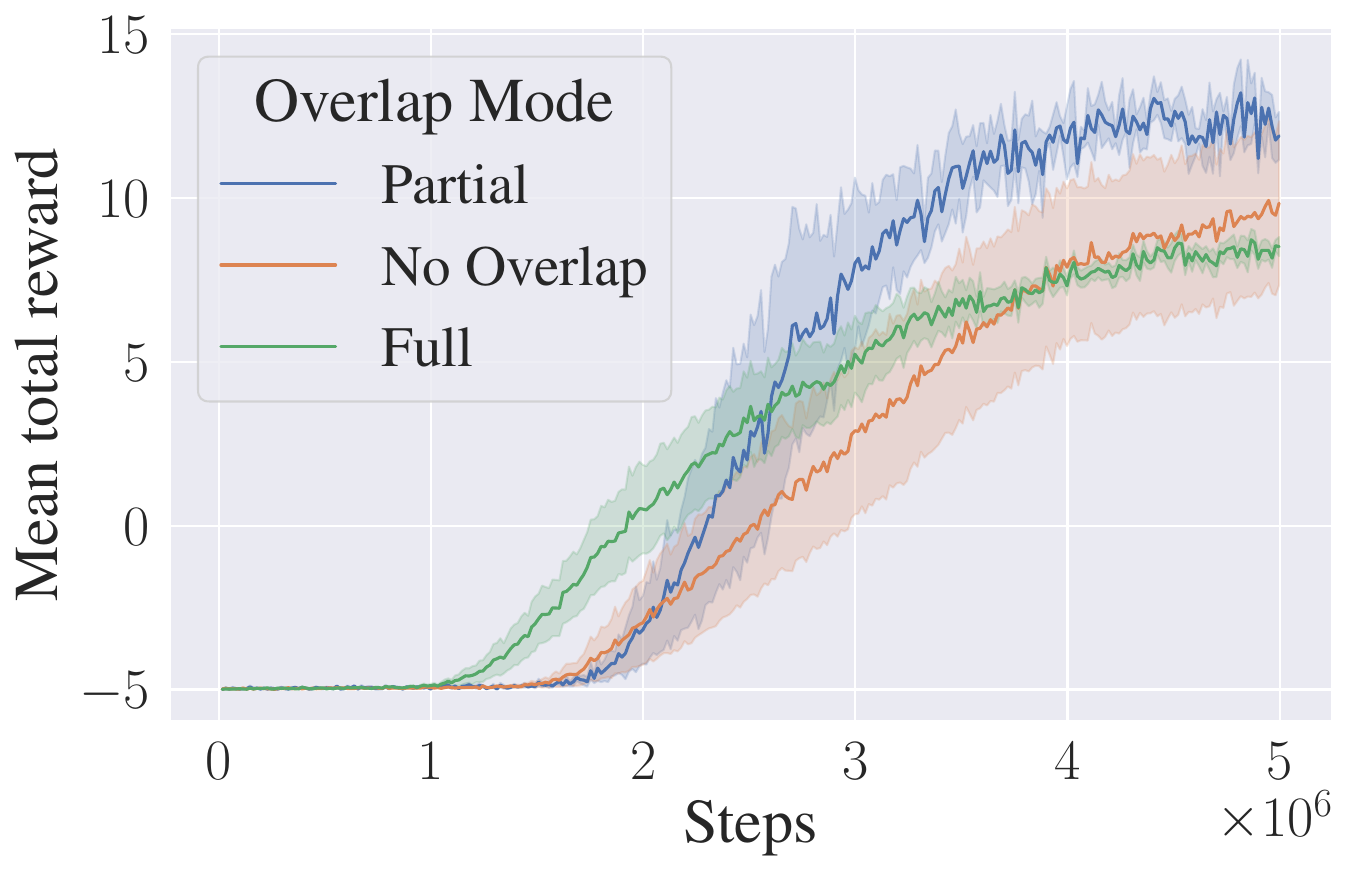}
        \caption{Mean reward of the MAPPO policy after learning\\from the given number of environment steps.}
        \label{fig:mappo_reward_curves}
    \end{subfigure}%
    \caption{Reward curves for Evolution (a) and DMARL (b). All curves are means from multiple seeds with standard error bars.}
    \vspace{-.25cm}
\end{figure*}
% \begin{figure*}
%     \centering
%     \captionsetup[subfigure]{justification=centering}
%     \begin{subfigure}{.3\linewidth}
%         \centering
%         \hspace{-1cm}
%         % \includegraphics[width=0.9\linewidth]{figures/arch.pdf}
%         \includegraphics[height=.65\linewidth]{figures/evo_reward_curves.pdf}
%         \caption{
%             Evolution reward curves.
%         }
%         \label{fig:evo_reward_curves}
%     \end{subfigure}%
%     % \hspace{.04\textwidth}
%     \begin{subfigure}{.3\linewidth}
%         \centering
%         \hspace{-1cm}
%         % \includegraphics[width=0.9\linewidth]{figures/arch.pdf}
%         \includegraphics[height=.65\linewidth]{figures/mappo_reward_curves.pdf}
%         \caption{RL reward curves.}
%         \label{fig:mappo_reward_curves}
%     \end{subfigure}%
%     % \hspace{.04\textwidth}
%     \begin{subfigure}{.3\textwidth}
%         \includegraphics[height=.65\linewidth]{figures/reward_and_mimicry_frequency.pdf}
%         \caption{}\label{fig:reward_and_mimicry_frequency}
%     \end{subfigure}%
%     \caption{Reward curves for Evolution (a) and DMARL (b) with different signal overlap settings. All curves are means averages from multiple seeds with standard error bars.}
% \end{figure*}

We can perform the same analysis for the uniformly random listener $\theta^s_U$.
The key difference is that in order to get the same behaviour, two of these kinds of mutations need to occur simultaneously.
One that shifts probability for the listener, and one for the speaker.
For the speaker, the probability needs to shift to any of the signals, and for the listener, the specific response to that signal needs to shift to the correct output. This probability must be less than the following:
\begin{align}
    \Bigg(p\frac{|\Sigma|}{|\mathcal{A}^s|}\Bigg)\cdot \Bigg(p\frac{1}{|\Sigma|}\cdot\frac{1}{|\mathcal{A}^l|} \Bigg) = \frac{p^2}{|\mathcal{A}^s||\mathcal{A}^l|}
\end{align}

Finally, we can tie this back to mimicry by reintroducing the externally generated useful signals from the previous section on independent optimisers.
Recall that $\mathrm{S}$ denotes the event that a signal received by the listener originates from an external source.
The key argument here is that for any $P(\mathrm{S})$, if optimisation is stuck in a non-communicative local optima, the listener will be optimised towards some $\theta^l_I$, provided $P(z|m,\mathrm{S})$ is deterministic.
As this happens, the chance for communicative behaviour from the speaker to emerge increases, for the reasons we have just seen.

% To formalise this line of thinking, when asking whether the inequality in Equation~\ref{eqn:local_noncomm_optima_gap} may hold for distribution of speakers $\theta^s_i$ and listener $\theta^l$, the answer depends on:
% \begin{align}
%     \alpha < \max_{i} P(a^l = z | a^s \in \Sigma, \theta^l,\theta^s_i)
% \end{align}
% Combining with Equation \ref{eqn:listener_eu_decomposed}, for the deterministic listener we know that $a^l = f(m)$, so:
% % \begin{align}
% %     \alpha < \max_i \sum_{m\in\mathcal{Z}}\frac{1}{|\mathcal{Z}|} P(a^s=m|a^s\in\Sigma,\theta^s_i)
% % \end{align}
% \begin{align}
%     \alpha < \max_{i} P(f(a^s) = z | a^s \in \Sigma, \theta^l_I,\theta^s_i)
% \end{align}
% Notice that as $P(f(a^s) = z |a^s\in\Sigma,\theta^s_i)$ is conditional the speaker sending a signal, therefore, $P(a^s\in\Sigma|\theta^s_i)$ may be arbitrarily low.
% Intuitively, the speaker does not need to communicate much, but when it does it needs 

\section{Experiments}

\subsection{Gridworld Environment}

To investigate the possible effects of mimicry on the emergence of communication, we introduce a simple cooperative gridworld environment.
In this environment, two agents must collect resources by simultaneously being on the same tile as the resource.
After which, a new resource is spawned at a random location.
The agents can be viewed as predators and the resources as prey, but we will not use this terminology as the resources are static (i.e. non-agentic).
The game is a sequential decision-making problem in which the goal is to collect as many resources as possible within 50 timesteps.
A reward of +10 is given when a resource is collected, and a penalty of -0.1 is given otherwise.

At each time step, each agent chooses between five spatial actions --- \textit{remain still}, \textit{go up}, \textit{go down}, \textit{go left}, or \textit{go right} --- and $|\Sigma_A|$ communicative actions.
If a spatial action is chosen, the agent moves accordingly.
If a communicative action is chosen, a signal is emitted from the agent's location that can be observed by the other agent in the next time step.
$\Sigma_A \subseteq \Sigma$ denotes the set of possible signals that an agent can emit.
Although there is no explicit cost to communication, a time step in which an agent emits a signal is a time step in which they are not moving towards a resource.
% So to maximise reward, agents should only communicate when necessary.

\textbf{Spatial Signals.}
As mentioned in the previous section, for mimicry to be possible the receiver cannot \textit{a priori} know the source of a given signal.
Therefore, systems such as dedicated communication channels are not applicable.
Instead, signals are emitted from a source at a given volume.
The volume dictates the number of tiles between a source and a receiver within which the signal is detectable.
In this environment, the \textit{model} that the \textit{mimic} will imitate is the resources.
Resources emit signals on each time step with probability $P_{res}$ and volume $v_{res}$ from a set $\Sigma_{res} \subseteq \Sigma$.
The agent signals and the resource signals may or may not overlap, and this parameter controls whether or not agents have the opportunity to evolve mimicry.

On that note, agents are equipped with four sensors that indicate which direction a signal comes from.
Each sensor is sensitive to signals emitted from a particular region of the world, relative to the position of the agent.
The agent can detect signals from the set $\Sigma$, so each sensor is a $|\Sigma|$-dimensional binary vector, where each entry indicates whether the corresponding signal has been emitted within the sensor's detection region.
Figure \ref{fig:signal-sensor-diagram-full-explanation} illustrates how each of the sensors works.
Formally, suppose that the agent is in position $a_x, a_y$ and the source emits the signal with index $m$, volume $v$, and from position $s_x, s_y$.
We denote each of the sensors as $\mathbf{v}_\uparrow, \mathbf{v}_\downarrow, \mathbf{v}_\rightarrow, \mathbf{v}_\leftarrow$, and let $d = |a_x - s_x| + |a_y - s_y|$.

\begin{itemize}[noitemsep]
    \item \textbf{Above sensor}: $\mathbf{v}_\uparrow[m] = 1$ if $s_y > a_y$ and $d < v$.
    \item \textbf{Below sensor}: $\mathbf{v}_\downarrow[m] = 1$ if $s_y < a_y$ and $d < v$.
    \item \textbf{Right sensor}: $\mathbf{v}_\rightarrow[m] = 1$ if $s_x > a_x$ and $d < v$.
    \item \textbf{Left sensor}: $\mathbf{v}_\leftarrow[m] = 1$ if $s_x < a_x$ and $d < v$.
\end{itemize}

When the information from the sensors is combined, the agent can distinguish signals coming from eight regions.
As the agent is not provided with information regarding the volume of the signal, it cannot distinguish between nearby sources with low volume and far-away sources with high volume.
When multiple detectable signals are made in the same region, they are also indistinguishable.
If they are made in different regions, then a unique pattern of sensor inputs will encode this information.
The `Example Situation' in Figure \ref{fig:signal-sensor-diagram-full-explanation} shows a case in which the agent in the middle (illustrated by the blue circle) is detecting two signals.
The first signal is emitted from the resource (illustrated by the golden star), indicated with the `A' is the 0th possible signal.
The second signal comes from the other agent to the right, which is emitting `B', i.e. the 1st signal.
We see that as the resource is located southwest of the agent, the 0th index of the below and left sensors.

\subsection{Experimental Results}

We conducted two sets of experiments in this gridworld environment.
Firstly, we evolved agents using the `SimpleGA' method \citep{such_deep_2017}, using the implementation in the evosax library \citep{lange_evosax_2023} and a population size of 256.
Secondly, we trained deep reinforcement learning agents with the MAPPO algorithm, using an adapted implementation from the JaxMARL library \citep{rutherford_jaxmarl_2023}, with a linearly annealed learning rate starting at $2\times10^{-3}$ and 128 parallel data collection environments.
Across all experiments, agents were parameterised by recurrent neural networks with four hidden layers of dimension 128, ReLU activations, and gated recurrent units \citep{cho_learning_2014}.

For our evolution experiments, we used a 5x5 grid for the world, with $|\Sigma_A|=5$ and $|\Sigma_{res}|=1$.
To investigate the effects of mimicable signals, we looked at the case where $\Sigma_{res} \subset \Sigma_A$, which we refer to as `signal overlap' in Figure~\ref{fig:evo_reward_curves}.
We ran 10 simulations with different seeds, with and without overlap, and found that for both cases 7 of the runs failed to converge.
Figure~\ref{fig:evo_reward_curves} only shows the 6 successful runs.
We see that the capacity for mimicry has a significant effect on the emergence of communication.

For the MAPPO experiments, we found that the same parameters as the evolution runs were much too easy.
So we increased the grid size to 10x10, and $|\Sigma_{res}|=5$.
As there were now more resource signals, we could test partial and full overlap settings.
For full overlap, $\Sigma_A = \Sigma_{res}$, and for partial $|\Sigma_A \cap \Sigma_{res}| = 1$.
Again, 10 seeds were run for each setting, with all runs being successful.
Figure~\ref{fig:mappo_reward_curves} shows these reward curves, and Table~\ref{tab:mappo_performance_stats} shows the overall final statistics for these runs.

\begin{table}
    \centering
    \begin{tabular}{c|c|c}
         \textbf{Overlap Mode} & \textbf{Mean Reward}  & \textbf{Std. Reward} \\
         \hline
         Full    & 8.51             & 0.91           \\
         Partial & \textbf{11.90}   & 1.26           \\
         None    & 9.84             & \textbf{7.92}  \\
    \end{tabular}
    \caption{MAPPO performance statistics by overlap mode, measured after learning from $5\times10^6$ environment steps.}
    \label{tab:mappo_performance_stats}
    \vspace{-.35cm}
\end{table}

Firstly, in this figure, we see that the `full overlap' setting leads to the earliest initial increases in performance.
Next, we see that although the partial setting starts increasing in performance at the same time as the `no overlap' setting, it increases much faster than the other two.
An important difference between the runs with and without mimicable signals is the variance in performance.
We see that mimicry significantly reduces the standard deviation -- with almost an order of magnitude difference between the `full' and `no' overlap cases.
Lastly, we see that the `full' overlap case performs worse than without overlap, on average by the end of the training, especially when comparing the highest-performing seeds.
This shows that in this environment, while anonymous signals may help initially, optimal performance involves being able to infer the source of the signals.
This points to a potential pitfall in this kind of set-up and helps explain the results in Figure~\ref{fig:reward_and_mimicry_frequency}.

In Figure~\ref{fig:reward_and_mimicry_frequency}, we are looking at data from MAPPO training runs with partial overlap.
Training curves from different seeds have been aligned by adjusting the iteration so they all pass through the origin at the same point.
This helps us see the similarities in the curve's shapes, without averaging out the patterns due to differences in when performance started to increase.
Alongside reward, on the right $y$-axis we have the \textit{mimicry frequency}.
This is the frequency in which a communication symbol sent by an agent is one of the mimicable symbols.
We see a sharp increase in this frequency at around the same point in which performance increases.
However, after performance begins to plateau, we see the mimicry frequency begin to drop.
This further demonstrates that while mimicry can help get communication off the ground, it may be detrimental when trying to refine the communicative strategies. 

% \begin{figure*}
%     \centering
%     \includegraphics[width=0.6\linewidth]{figures/reward_and_mimicry_frequency.pdf}
%     \caption{}\label{fig:reward_and_mimicry_frequency}
% \end{figure*}

\begin{figure}
    \centering
    \includegraphics[width=.9\linewidth]{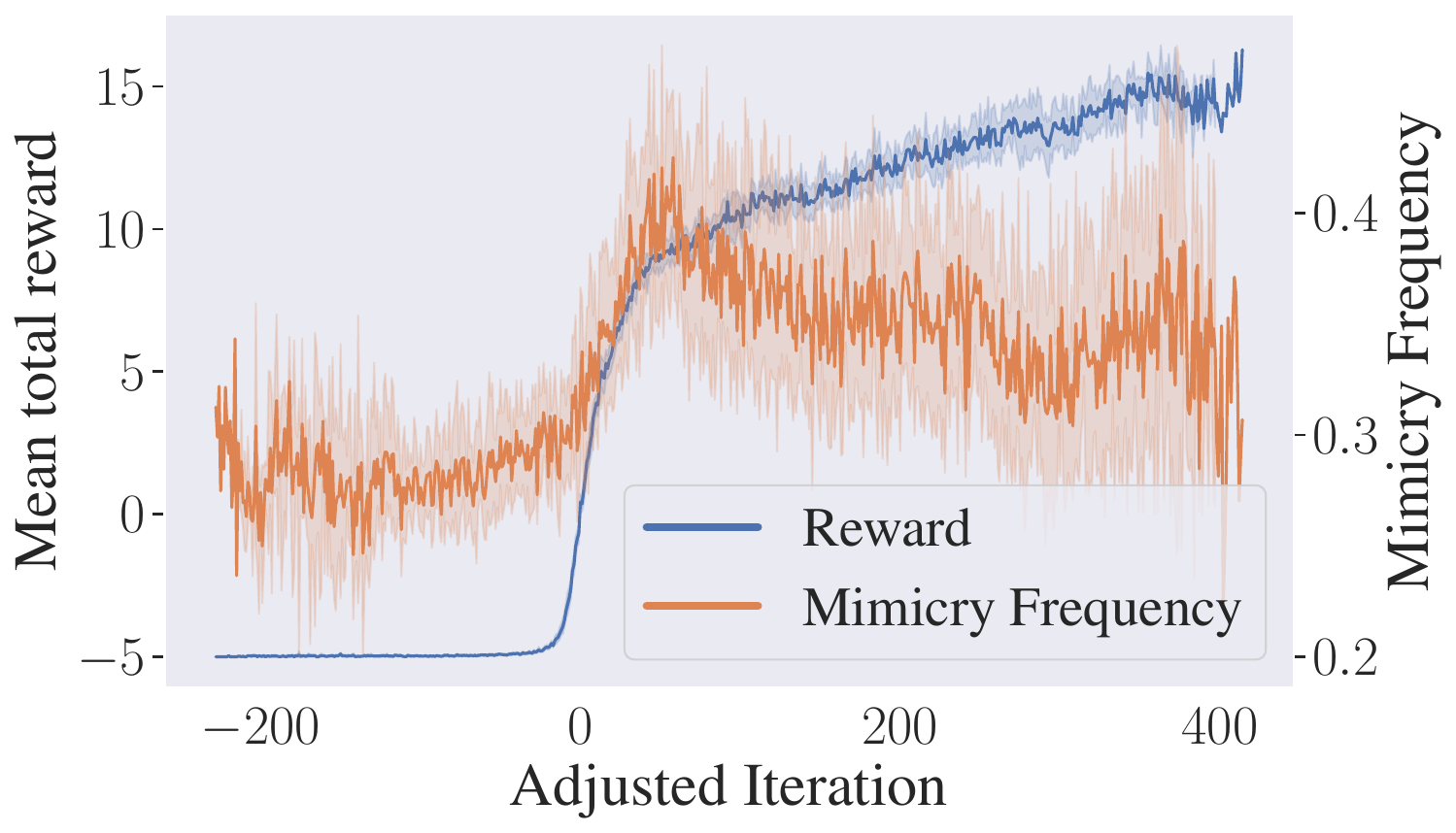}
    \caption{
        Partial overlap MAPPO reward curves aligned to the iteration in which the mean total reward exceeded 0, plotted alongside the frequency in which agents mimicked the `externally generated signals', i.e. used resource signals.
    }\label{fig:reward_and_mimicry_frequency}
    \vspace{-.25cm}
\end{figure}

\section{Conclusions}

In this paper, we have investigated the relationship between the emergence of cooperative communication, and the capacity for speakers to mimic preexisting useful signals.
We have demonstrated, analytically and empirically, that this may benefit the optimisation dynamics of populations of agents learning or evolving.
However, there are still many questions left unanswered by this analysis.
While we looked at how mimicry affects non-communicative local optima, we did not yet analytically analyse the effects when agents are randomly initialised.
On the other hand, we have not yet fully explored some of the negative side effects of this mimicry.
Although our results from Figures~\ref{fig:mappo_reward_curves} and \ref{fig:reward_and_mimicry_frequency} suggest that strategies that require disambiguation between the sources of signals may be a good place to start.

\footnotesize
\bibliographystyle{apalike}
\bibliography{references}

\begin{thebibliography}{}

\bibitem[Ackley and Littman, 1994]{ackley_altruism_1994}
Ackley, D.~H. and Littman, M.~L. (1994).
\newblock Altruism in the {Evolution} of {Communication}.
\newblock In {\em Artificial {Life} {IV}: {Proceedings} of the {Fourth} {International} {Workshop} on the {Synthesis} and {Simulation} of {Living} {Systems}}. MIT Press.

\bibitem[Angeline et~al., 1994]{angeline_evolutionary_1994}
Angeline, P., Saunders, G., and Pollack, J. (1994).
\newblock An evolutionary algorithm that constructs recurrent neural networks.
\newblock {\em IEEE Transactions on Neural Networks}, 5(1):54--65.

\bibitem[Bates, 1862]{bates_contributions_1862}
Bates, H.~W. (1862).
\newblock Contributions to an insect fauna of the {Amazon} valley ({Lepidoptera}: {Heliconidae}).
\newblock {\em Transactions of the Linnean Society of London}, 23:495--56.

\bibitem[Bullock, 1997]{bullock_evolutionary_1997}
Bullock, S. (1997).
\newblock {\em Evolutionary {Simulation} {Models}: {On} {Their} {Character}, and {Application} to {Problems} {Concerning} the {Evolution} of {Natural} {Signalling} {Systems}}.
\newblock {PhD}, University of Sussex.
\newblock Publisher: University of Sussex.

\bibitem[Cho et~al., 2014]{cho_learning_2014}
Cho, K., van Merriënboer, B., Gulcehre, C., Bahdanau, D., Bougares, F., Schwenk, H., and Bengio, Y. (2014).
\newblock Learning {Phrase} {Representations} using {RNN} {Encoder}–{Decoder} for {Statistical} {Machine} {Translation}.
\newblock In Moschitti, A., Pang, B., and Daelemans, W., editors, {\em Proceedings of the 2014 {Conference} on {Empirical} {Methods} in {Natural} {Language} {Processing} ({EMNLP})}, pages 1724--1734, Doha, Qatar. Association for Computational Linguistics.

\bibitem[Conti et~al., 2018]{conti_improving_2018}
Conti, E., Madhavan, V., Such, F.~P., Lehman, J., Stanley, K.~O., and Clune, J. (2018).
\newblock Improving exploration in evolution strategies for deep reinforcement learning via a population of novelty-seeking agents.
\newblock In {\em Proceedings of the 32nd {International} {Conference} on {Neural} {Information} {Processing} {Systems}}, {NIPS}'18, pages 5032--5043, Red Hook, NY, USA. Curran Associates Inc.

\bibitem[Floreano and Claudio, 2008]{floreano_bio-inspired_2008}
Floreano, D. and Claudio, M. (2008).
\newblock {\em Bio-{Inspired} {Artificial} {Intelligence}}.
\newblock MIT Press.

\bibitem[Floreano et~al., 2007]{floreano_evolutionary_2007}
Floreano, D., Mitri, S., Magnenat, S., and Keller, L. (2007).
\newblock Evolutionary conditions for the emergence of communication in robots.
\newblock {\em Current biology: CB}, 17(6):514--519.

\bibitem[Foerster et~al., 2016]{foerster_learning_2016}
Foerster, J., Assael, I.~A., Freitas, N.~d., and Whiteson, S. (2016).
\newblock Learning to {Communicate} with {Deep} {Multi}-{Agent} {Reinforcement} {Learning}.
\newblock In {D. D. Lee and M. Sugiyama and U. V. Luxburg and I. Guyon and R. Garnett}, editor, {\em Advances in {Neural} {Information} {Processing} {Systems} 29}, pages 2137--2145. Curran Associates, Inc.

\bibitem[Fox and Bullock, 2023]{fox_nectar_2023}
Fox, R. and Bullock, S. (2023).
\newblock Nectar of the {Bots}: {Evolving} {Bidirectional} {Referential} {Communication}.
\newblock {\em Adaptive Behavior}, 31(1):65--86.
\newblock Publisher: SAGE Publications Ltd STM.

\bibitem[Galván, 2021]{galvan_neuroevolution_2021}
Galván, E. (2021).
\newblock Neuroevolution in deep neural networks: a comprehensive survey.
\newblock {\em ACM SIGEVOlution}, 14(1):3--7.

\bibitem[Goldman and Zilberstein, 2004]{goldman_decentralized_2004}
Goldman, C.~V. and Zilberstein, S. (2004).
\newblock Decentralized control of cooperative systems: categorization and complexity analysis.
\newblock {\em Journal of Artificial Intelligence Research}, 22(1):143--174.

\bibitem[Goldman and Zilberstein, 2008]{goldman_communication-based_2008}
Goldman, C.~V. and Zilberstein, S. (2008).
\newblock Communication-{Based} {Decomposition} {Mechanisms} for {Decentralized} {MDPs}.
\newblock {\em Journal of Artificial Intelligence Research}, 32:169--202.
\newblock arXiv:1111.0065 [cs].

\bibitem[Grice, 1975]{grice_logic_1975}
Grice, H.~P. (1975).
\newblock Logic and {Conversation}.
\newblock {\em Syntax and Semantics}, pages 41--58.

\bibitem[Hessel et~al., 2018]{hessel_rainbow_2018}
Hessel, M., Modayil, J., Hasselt, H.~v., Schaul, T., Ostrovski, G., Dabney, W., Horgan, D., Piot, B., Azar, M., and Silver, D. (2018).
\newblock Rainbow: {Combining} {Improvements} in {Deep} {Reinforcement} {Learning}.
\newblock {\em Proceedings of the AAAI Conference on Artificial Intelligence}, 32(1).
\newblock Number: 1.

\bibitem[Holland, 1992]{holland_genetic_1992}
Holland, J.~H. (1992).
\newblock Genetic {Algorithms}.
\newblock {\em Scientific American}, 267(1):66--73.

\bibitem[Hraber et~al., 1997]{hraber_ecology_1997}
Hraber, P.~T., Jones, T., and Forrest, S. (1997).
\newblock The {Ecology} of {Echo}.
\newblock {\em Artificial Life}, 3(3):165--190.

\bibitem[Hussain and Belardinelli, 2024]{hussain_stability_2024}
Hussain, A. and Belardinelli, F. (2024).
\newblock Stability of {Multi}-{Agent} {Learning} in {Competitive} {Networks}: {Delaying} the {Onset} of {Chaos}.
\newblock {\em Proceedings of the AAAI Conference on Artificial Intelligence}, 38(16):17435--17443.
\newblock Number: 16.

\bibitem[Hussain et~al., 2023]{hussain_asymptotic_2023}
Hussain, A.~A., Belardinelli, F., and Piliouras, G. (2023).
\newblock Asymptotic {Convergence} and {Performance} of {Multi}-{Agent} {Q}-{Learning} {Dynamics}.
\newblock In {\em 2023 {International} {Conference} on {Autonomous} {Agents} and {Multiagent} {Systems} ({AAMAS} 23)}, London, UK. International Foundation for Autonomous Agents and Multiagent Systems (IFAAMAS).

\bibitem[Islam and Grogono, 2016]{islam_modeling_2016}
Islam, M. and Grogono, P. (2016).
\newblock Modeling the {Evolution} of {Mimicry}.
\newblock In {\em Proceedings of the {Artificial} {Life} {Conference}}, pages 442--449. MIT Press.

\bibitem[Jaderberg et~al., 2019]{jaderberg_human-level_2019}
Jaderberg, M., Czarnecki, W.~M., Dunning, I., Marris, L., Lever, G., Castañeda, A.~G., Beattie, C., Rabinowitz, N.~C., Morcos, A.~S., Ruderman, A., Sonnerat, N., Green, T., Deason, L., Leibo, J.~Z., Silver, D., Hassabis, D., Kavukcuoglu, K., and Graepel, T. (2019).
\newblock Human-level performance in {3D} multiplayer games with population-based reinforcement learning.
\newblock {\em Science}, 364(6443):859--865.
\newblock arXiv: 1807.01281 Publisher: American Association for the Advancement of Science.

\bibitem[Kleisner and Maran, 2019]{kleisner_introduction_2019}
Kleisner, K. and Maran, T. (2019).
\newblock Introduction to {Signs} and {Communication} in {Mimicry}.
\newblock {\em Biosemiotics}, 12:1--6.

\bibitem[Lange, 2023]{lange_evosax_2023}
Lange, R.~T. (2023).
\newblock evosax: {JAX}-{Based} {Evolution} {Strategies}.
\newblock In {\em Proceedings of the {Companion} {Conference} on {Genetic} and {Evolutionary} {Computation}}, {GECCO} '23 {Companion}, pages 659--662, New York, NY, USA. Association for Computing Machinery.

\bibitem[Lehman et~al., 2018]{lehman_surprising_2018}
Lehman, J., Clune, J., Misevic, D., Adami, C., Altenberg, L., Beaulieu, J., Bentley, P.~J., Bernard, S., Beslon, G., Bryson, D.~M., Chrabaszcz, P., Cheney, N., Cully, A., Doncieux, S., Dyer, F.~C., Ellefsen, K.~O., Feldt, R., Fischer, S., Forrest, S., Frénoy, A., Gagné, C., Goff, L.~L., Grabowski, L.~M., Hodjat, B., Hutter, F., Keller, L., Knibbe, C., Krcah, P., Lenski, R.~E., Lipson, H., MacCurdy, R., Maestre, C., Miikkulainen, R., Mitri, S., Moriarty, D.~E., Mouret, J.-B., Nguyen, A., Ofria, C., Parizeau, M., Parsons, D., Pennock, R.~T., Punch, W.~F., Ray, T.~S., Schoenauer, M., Shulte, E., Sims, K., Stanley, K.~O., Taddei, F., Tarapore, D., Thibault, S., Weimer, W., Watson, R., and Yosinski, J. (2018).
\newblock The {Surprising} {Creativity} of {Digital} {Evolution}: {A} {Collection} of {Anecdotes} from the {Evolutionary} {Computation} and {Artificial} {Life} {Research} {Communities}.
\newblock arXiv: 1803.03453.

\bibitem[Lillicrap et~al., 2016]{lillicrap_continuous_2016}
Lillicrap, T.~P., Hunt, J.~J., Pritzel, A., Heess, N., Erez, T., Tassa, Y., Silver, D., and Wierstra, D. (2016).
\newblock Continuous control with deep reinforcement learning.
\newblock In Bengio, Y. and LeCun, Y., editors, {\em 4th {International} {Conference} on {Learning} {Representations}, {ICLR} 2016, {San} {Juan}, {Puerto} {Rico}, {May} 2-4, 2016, {Conference} {Track} {Proceedings}}.

\bibitem[Maran, 2017]{maran_mimicry_2017}
Maran, T. (2017).
\newblock {\em Mimicry and {Meaning}: {Structure} and {Semiotics} of {Biological} {Mimicry}}, volume~16 of {\em Biosemiotics}.
\newblock Springer International Publishing, Cham.

\bibitem[Marriott et~al., 2018]{marriott_social_2018}
Marriott, C., Borg, J.~M., Andras, P., and Smaldino, P.~E. (2018).
\newblock Social {Learning} and {Cultural} {Evolution} in {Artificial} {Life}.
\newblock {\em Artificial Life}, 24(1):5--9.

\bibitem[Mirolli and Nolfi, 2010]{mirolli_evolving_2010}
Mirolli, M. and Nolfi, S. (2010).
\newblock Evolving {Communication} in {Embodied} {Agents}: {Theory}, {Methods}, and {Evaluation}.
\newblock In Nolfi, S. and Mirolli, M., editors, {\em Evolution of {Communication} and {Language} in {Embodied} {Agents}}, pages 105--121. Springer, Berlin, Heidelberg.

\bibitem[Mitchell, 1996]{mitchell_introduction_1996}
Mitchell, M. (1996).
\newblock {\em An {Introduction} to {Genetic} {Algorithms}}.
\newblock MIT Press.

\bibitem[Mnih et~al., 2015]{mnih_human-level_2015}
Mnih, V., Kavukcuoglu, K., Silver, D., Rusu, A.~A., Veness, J., Bellemare, M.~G., Graves, A., Riedmiller, M., Fidjeland, A.~K., Ostrovski, G., Petersen, S., Beattie, C., Sadik, A., Antonoglou, I., King, H., Kumaran, D., Wierstra, D., Legg, S., and Hassabis, D. (2015).
\newblock Human-level control through deep reinforcement learning.
\newblock {\em Nature}, 518(7540):529.

\bibitem[Noble et~al., 2002]{noble_adaptive_2002}
Noble, J., Di~Paolo, E.~A., and Bullock, S. (2002).
\newblock Adaptive {Factors} in the {Evolution} of {Signaling} {Systems}.
\newblock In Cangelosi, A. and Parisi, D., editors, {\em Simulating the {Evolution} of {Language}}, pages 53--77. Springer, London.

\bibitem[Oliehoek and Amato, 2016]{oliehoek_concise_2016}
Oliehoek, F.~A. and Amato, C. (2016).
\newblock {\em A {Concise} {Introduction} to {Decentralized} {POMDPs}}.
\newblock Springer International Publishing, Cham.
\newblock Series Title: SpringerBriefs in Intelligent Systems.

\bibitem[Oliphant, 1996]{oliphant_dilemma_1996}
Oliphant, M. (1996).
\newblock The dilemma of {Saussurean} communication.
\newblock {\em Biosystems}, 37(1):31--38.

\bibitem[O'Reilly et~al., 2019]{oreilly_deaf_2019}
O'Reilly, L.~J., Agassiz, D. J.~L., Neil, T.~R., and Holderied, M.~W. (2019).
\newblock Deaf moths employ acoustic {Müllerian} mimicry against bats using wingbeat-powered tymbals.
\newblock {\em Scientific Reports}, 9.

\bibitem[Parisi, 1997]{parisi_artificial_1997}
Parisi, D. (1997).
\newblock An artificial life approach to language.
\newblock {\em Brain and Language}, 59(1):121--146.

\bibitem[Quicke, 2017]{quicke_mimicry_2017}
Quicke, D. L.~J. (2017).
\newblock {\em Mimicry, {Crypsis}, {Masquerade} and other {Adaptive} {Resemblances}}.
\newblock Wiley-Blackwell.

\bibitem[Reynolds, 2011]{reynolds_interactive_2011}
Reynolds, C. (2011).
\newblock Interactive {Evolution} of {Camouflage}.
\newblock {\em Artificial Life}, 17(2):123--136.

\bibitem[Ronald and Schoenauer, 1994]{ronald_genetic_1994}
Ronald, E. and Schoenauer, M. (1994).
\newblock Genetic lander: {An} experiment in accurate neuro-genetic control.
\newblock In Davidor, Y., Schwefel, H.-P., and Männer, R., editors, {\em Parallel {Problem} {Solving} from {Nature} — {PPSN} {III}}, pages 452--461, Berlin, Heidelberg. Springer.

\bibitem[Rutherford et~al., 2023]{rutherford_jaxmarl_2023}
Rutherford, A., Ellis, B., Gallici, M., Cook, J., Lupu, A., Ingvarsson, G., Willi, T., Khan, A., de~Witt, C.~S., Souly, A., Bandyopadhyay, S., Samvelyan, M., Jiang, M., Lange, R.~T., Whiteson, S., Lacerda, B., Hawes, N., Rocktaschel, T., Lu, C., and Foerster, J.~N. (2023).
\newblock {JaxMARL}: {Multi}-{Agent} {RL} {Environments} in {JAX}.
\newblock In {\em Agent {Learning} in {Open}-{Endedness} ({ALOE}) {Workshop} at {NeurIPS} 2023}, New Orleans. arXiv.
\newblock arXiv:2311.10090 [cs].

\bibitem[Schulman et~al., 2015]{schulman_trust_2015}
Schulman, J., Levine, S., Abbeel, P., Jordan, M., and Moritz, P. (2015).
\newblock Trust {Region} {Policy} {Optimization}.
\newblock In {\em Proceedings of the 32nd {International} {Conference} on {Machine} {Learning}}, pages 1889--1897. PMLR.
\newblock ISSN: 1938-7228.

\bibitem[Schulman et~al., 2017]{schulman_proximal_2017}
Schulman, J., Wolski, F., Dhariwal, P., Radford, A., and Klimov, O. (2017).
\newblock Proximal {Policy} {Optimization} {Algorithms}.
\newblock arXiv:1707.06347 [cs].

\bibitem[Stanley and Miikkulainen, 2002]{stanley_evolving_2002}
Stanley, K.~O. and Miikkulainen, R. (2002).
\newblock Evolving {Neural} {Networks} through {Augmenting} {Topologies}.
\newblock {\em Evolutionary Computation}, 10(2):99--127.

\bibitem[Such et~al., 2017]{such_deep_2017}
Such, F.~P., Madhavan, V., Conti, E., Lehman, J., Stanley, K.~O., and Clune, J. (2017).
\newblock Deep {Neuroevolution}: {Genetic} {Algorithms} {Are} a {Competitive} {Alternative} for {Training} {Deep} {Neural} {Networks} for {Reinforcement} {Learning}.
\newblock arXiv: 1712.06567.

\bibitem[Sukhbaatar et~al., 2016]{sukhbaatar_learning_2016}
Sukhbaatar, S., Szlam, A., and Fergus, R. (2016).
\newblock Learning multiagent communication with backpropagation.
\newblock In {\em Proceedings of the 30th {International} {Conference} on {Neural} {Information} {Processing} {Systems}}, pages 2252--2260, Barcelona, Spain. Neural Information Processing Systems.

\bibitem[Sutton and Barto, 1998]{sutton_reinforcement_1998}
Sutton, R. and Barto, A. (1998).
\newblock Reinforcement {Learning}: {An} {Introduction}.
\newblock {\em IEEE Transactions on Neural Networks}, 9(5):1054--1054.

\bibitem[Wang et~al., 2017]{wang_sample_2017}
Wang, Z., Bapst, V., Heess, N., Mnih, V., Munos, R., Kavukcuoglu, K., and Freitas, N.~d. (2017).
\newblock Sample {Efficient} {Actor}-{Critic} with {Experience} {Replay}.

\bibitem[Wickler, 1965]{wickler_mimicry_1965}
Wickler, W. (1965).
\newblock Mimicry and the {Evolution} of {Animal} {Communication}.
\newblock {\em Nature}, 208(5010):519--521.

\bibitem[Yeargan and Quate, 1996]{yeargan_juvenile_1996}
Yeargan, K.~V. and Quate, L.~W. (1996).
\newblock Juvenile bolas spiders attract psychodid flies.
\newblock {\em Oecologia}, 106(2):266--271.

\bibitem[Yu et~al., 2022]{yu_surprising_2022}
Yu, C., Velu, A., Vinitsky, E., Gao, J., Wang, Y., Bayen, A., and Wu, Y. (2022).
\newblock The {Surprising} {Effectiveness} of {PPO} in {Cooperative} {Multi}-{Agent} {Games}.
\newblock In {\em 36th {Conference} on {Neural} {Information} {Processing} {Systems} ({NeurIPS} 2022) {Track} on {Datasets} and {Benchmarks}.}, New Orleans, USA.

\end{thebibliography}

\end{document}